\documentclass[12pt, final]{l4dc2023}

\usepackage{times}
\usepackage{bm}
\usepackage{mathtools}

\newtheorem{assumption}{Assumption}

\newcommand{\W}{\mathbf{w}}

\newcommand{\T}{^\top}

\newcommand{\Pz}{P_{\mathrm{init}}}
\newcommand{\Z}{\mathbf{z}}

\newcommand{\V}{\mathbf{v}}
\newcommand{\U}{\mathbf{u}}

\newcommand{\X}{\mathbf{x}}

\newcommand{\defmath}{\vcentcolon =}
\newcommand{\x}{\mathrm{x}}
\renewcommand{\u}{\mathrm{u}}
\newcommand{\w}{\mathrm{w}}
\newcommand{\f}{\mathrm{f}}
\usepackage{enumitem}
\usepackage{booktabs}
\usepackage[capitalise]{acro}
\usepackage{hyperref}

\DeclareAcronym{RMPC}{
  short = RMPC,
  long  = robust MPC,
}
\DeclareAcronym{SLS}{
  short = SLS,
  long  = system level synthesis,
}
\DeclareAcronym{MPC}{
  short = MPC,
  long  = model predictive control,
}

\DeclareAcronym{LTV}{
  short = LTV,
  long  = linear time-varying,
}

\DeclareAcronym{RPI}{
  short = RPI,
  long  = robust positive invariant,
}

\DeclareAcronym{QP}{
  short = QP,
  long  = quadratic program,
}

\DeclareAcronym{CBF}{
  short = CBF,
  long  = control barrier function,
  long-plural-form = control barrier functions,
}

\DeclareAcronym{MPSF}{
  short = MPSF,
  long  = model predictive safety filter,
  long-plural-form = model predictive safety filters,
}

\DeclareAcronym{RCI}{
  short = RCI,
  long  = robust control invariant,
}

\title[Predictive safety filter using system level synthesis]{Predictive safety filter using system level synthesis}\usepackage{times}

\author{\Name{Antoine P. Leeman}$\textsuperscript{\dag}$\Email{aleeman@ethz.ch}\thanks{This work has been supported by the European Space Agency under OSIP 4000133352, the Swiss Space Center, and the Swiss National Science Foundation under NCCR Automation (grant agreement 51NF40 180545).}\\
\Name{Johannes K\"ohler}$\textsuperscript{\dag}$ \Email{jkoehle@ethz.ch}\\
\Name{Samir Bennani}$\textsuperscript{\ddag}$ \Email{Samir.Bennani@esa.int}\\
\Name{Melanie N. Zeilinger}$\textsuperscript{\dag}$ \Email{mzeilinger@ethz.ch}\\
\addr $\textsuperscript{\dag}$Institute for Dynamic Systems and Control, ETH Zurich, Switzerland\\
\addr $\textsuperscript{\ddag}$ESTEC, European Space Agency, the Netherlands
}

\begin{document}

\maketitle
\begin{abstract}
Safety filters provide modular techniques to augment {potentially} unsafe control inputs (e.g. from learning-based controllers or humans) with safety guarantees in the form of constraint satisfaction. 
In this paper, we present an improved model predictive safety filter (MPSF) formulation, which incorporates system level synthesis techniques in the design. 
The resulting SL-MPSF scheme ensures safety for linear systems subject to bounded disturbances in an enlarged safe set. It requires less severe and frequent modifications of potentially unsafe control inputs compared to existing MPSF formulations to certify safety. 
In addition, we propose an explicit variant of the SL-MPSF formulation, which maintains scalability, and reduces the required online computational effort - the main drawback of the MPSF. 
The benefits of the proposed system level safety filter formulations compared to state-of-the-art MPSF formulations are demonstrated using a numerical example. 
\end{abstract}
\begin{keywords}%
Safety Filter, Safety certification, Model Predictive Control, System Level Synthesis
\end{keywords}
\section{Introduction}
Learning-based controllers have demonstrated high performance for control of uncertain systems in complex environments, see, e.g. ~\cite{Hwangbo2019, Ibarz2021, mnih2015human} and references therein. However, their potential can still not be fully exploited in many industrial applications due to the lack of safety guarantees {i.e., in terms of constraint satisfaction}.

Safety filters are modular techniques that can be combined with any controller and modify{/filter} a potentially unsafe input as little as needed while ensuring that the state remains in some safe set.
Methods based on Hamilton-Jacobi reachability analysis~\cite{Gillula2012,Fisac2019ASystems} provide a general approach for computing safe sets. However, their practical applicability is often limited due to scalability issues associated with solving nonlinear partial differential equations. 
\Acp{CBF} provide a popular alternative to define safe sets~\cite{Ames2019}, but their design commonly requires system-specific intuition related to the choice of the control Lyapunov function. 
\Acp{MPSF}~\cite{Wabersich2018LinearControl,Wabersich2021ASystems} provide a flexible way to implicitly define a safe set and a safety filter as the solution to an optimisation problem, at the cost of increased online computational effort. 

Overall, these three main approaches used to define safety filters offer different trade-offs between offline design complexity, online computational demand, and conservativeness in the design. 
Conservativeness is reflected by the safe set size and the amount of control intervention. Maximising the size of the safe set increases the operational range of the system and enables, e.g. more extensive (safe) exploration particularly in the learning setting. 
In addition, safety filter interventions affect the performance of the learning-based (or the human) controller, such as stability~\mbox{\cite{dai2021lyapunov}}, operational cost, or tracking error.
This paper presents an \ac{MPSF} approach that reduces conservativeness compared with existing \ac{MPSF} {by increasing the size of the safe set, and reducing the required safety filter interventions,} and thereby enables higher performance for problems with larger model uncertainty.

The main concept of the \ac{MPSF} is based on \ac{MPC}: a safe backup trajectory is predicted, ensuring the system could be steered to a safe (invariant) set, while minimally modifying the potentially unsafe input $u_{\mathcal{L}}$. The safe set and safe backup strategy in \ac{MPSF} schemes are thereby characterised by the solution to an optimisation problem. A key challenge is the computation of safe backup trajectories in the presence of disturbances and model uncertainties, which are particularly relevant in the {learning setting}. The first \ac{MPSF} was presented in~\cite{Wabersich2018LinearControl}, and exploits the linear \ac{RMPC}  from~\cite{Mayne2005RobustDisturbances}. This formulation was later extended to distributed systems~\cite{MUNTWILER20205258}, parametric uncertainties~\cite{Didier2021AdaptiveControl}, nonlinear systems~\cite{Wabersich2021ASystems}, and stochastic noise~\cite{Wabersich2022ProbabilisticControl}.
Notably, all of these \ac{MPSF} formulations use {an offline fixed auxiliary tube-controller} to mitigate the effect of disturbances/uncertainties. While the tube-controller allows for a scalable and efficient offline design, this simplification can result in overly cautious \ac{MPSF} schemes, limiting their practical performance.

\paragraph{Contribution} In this paper, we propose a novel \ac{MPSF} formulation focused on a robust design for linear systems, that addresses the main limitation of the fixed tube-controller, and hence significantly mitigates the conservativeness of existing {\ac{MPSF}} schemes.
To achieve this, we combine \ac{MPSF} with \ac{SLS}~\cite{Anderson2019} to create a safety filter that ensures safety with reduced control interventions.
In particular, \ac{SLS} enables optimisation over affine feedback policies in \ac{RMPC}~\cite{Chen2021SystemApproximation,Sieber2021SystemMPC}, compare also disturbance feedback \ac{MPC}~\cite{Goulart2006OptimizationConstraints} (Section~\ref{sec:paramerisation}). As a result, the proposed system level \ac{MPSF} (SL-\ac{MPSF}) overcomes the limitation of the fixed tube-controller in existing MSPF schemes, and ensures a larger {safe set}, {and lower control intervention} at the cost of increased computational complexity. The proposed flexible SL-\ac{MPSF} hence simultaneously optimises the safe backup trajectory and tube-controller (Section~\ref{sec:sl-mpc}). 
Notably, the presented SL-\ac{MPSF} formulation uses a general terminal set to ensure recursive feasibility and does not require the finite-impulse-response constraint from~\cite{Sieber2021SystemMPC}, which constitutes an alternative \ac{RMPC} formulation based on \ac{SLS}. 

The second contribution is an explicit safety filter that leverages the flexibility of \ac{SLS} to design an explicit safe set and an explicit backup control law offline (Section \ref{sec:explicit}). The resulting scheme is more conservative than SL-\ac{MPSF} but {requires no online optimisation}.
The explicit safe set is constructed by imposing the disturbance reachable set to safely return to the safe set within the prediction horizon. The safe set is defined using a simple shape and is hence easy to implement and scalable. 
Similar methods are used in~\cite{Parsi2022ComputationallyTightening} and \cite[Section 4.2]{Sieber2021SystemMPC} to compute feedback policies and disturbance reachable sets offline, however, they still require online optimisation to ensure robust constraint satisfaction.

Both safety filters proposed in this paper offer a different trade-off between conservativeness, scalability, offline design complexity and online computation demand compared with available techniques and thereby enlarge the scope of this promising safety approach.
The benefits of both methods are illustrated by comparing them against state-of-the-art safety filters using a numerical example (Section~\ref{sec:num}).

\paragraph{Notation} The set of non-negative integers is given by $\mathbb{N}_{\ge 0}$. For a vector $x\in \mathbb{R}^n$, we denote the $p$-norm with $p \in \{1,2,\infty\}$ by $\|x\|_p$. The infinity-norm ball is given by $\mathcal{B}_\infty^n \defmath \{w\in \mathbb{R}^n | \|w\|_\infty \le 1 \}$. 
For two sets $\mathcal{W}_1$ and $\mathcal{W}_2$, the Minkowski sum is defined as $\mathcal{W}_1 \oplus \mathcal{W}_2 \defmath  \{ w_1 + w_2 |~w_1 \in \mathcal{W}_1, w_2 \in \mathcal{W}_2\}$.
{We also define $\mathcal{W}^k$ as the Cartesian product of $\mathcal{W}$ with itself $k$ times, i.e., $\mathcal{W}^k = \mathcal{W} \times \dots \times \mathcal{W}$ ($k$ times).}
The identity matrix is denoted by $\mathcal{I}_{n}\in\mathbb{R}^{n\times n}$ and the dimensions are omitted if they can be inferred from the context. 
Let $0_{p,q}\in \mathbb{R}^{p\times q}$ and $0_{n}\in \mathbb{R}^{n}$ be respectively a matrix and a vector of zeros.
We denote stacked vectors by $(a,b) = [a\T~b\T]\T$.
The block diagonal matrix consisting of matrices $A_1,\dots,A_T$ is denoted by $\mathcal{A}=\mathrm{diag}(A_1,\dots,A_T)$.
Let $\mathcal{L}^{N,p\times q}$ denote the set of all block lower-triangular matrices with the following structure%
\begin{equation}
     M = \begin{bmatrix} M^{0,0} & 0_{p,q} & \dots & 0_{p,q} \\ M^{1,1} & M^{1,0} & \dots & 0_{p,q} \\ \vdots & \vdots & \ddots & \vdots \\M^{N,N} & M^{N,N-1} & \dots & M^{N,0} \end{bmatrix},
\nonumber%
 \end{equation}
where $M^{i,j}\in \mathbb{R}^{p\times q}$. We denote the $k^\text{th}$ block row of $M$ as $M^k\defmath [M^{k,k} \dots~M^{k,0},0_{p,q(N-1-k)}]$.
\section{Preliminaries}
We consider a linear system%
\begin{equation}
    x(t+1) = A x(t) + B u(t) + B_\w w(t),~t\in \mathbb{N}_{\ge 0},
    \label{eq:lin}
\end{equation}
with state $x(t)\in\mathbb{R}^n$, control input $u(t)\in\mathbb{R}^m$ and unknown disturbance $w(t)\in B_\textrm{w} \mathcal{W}\subset \mathbb{R}^n$, at discrete time instances $t\in\mathbb{N}_{\ge 0}$. For simplicity of exposition, we consider a box $ \mathcal{W} \defmath \mathcal{B}_\infty^{n}$, and assume that $B_\w$ is an invertible square matrix, compare also Remark~\ref{rem:ic_dist_size} below.
The system is {safe if it satisfies the} following state and input constraints%
\begin{equation*}
\begin{aligned}
     x(t) \in \mathcal{X} \defmath \{ x|~A_{\x,j}x \le b_{\x,j},~j=1,\dots, n_\x\},~
     u(t) \in \mathcal{U} \defmath \{ u|~A_{\u,j}u \le b_{\u,j},~j=1,\dots, n_\u\},
 \end{aligned}
\end{equation*}
 for all $t\in\mathbb{N}_{\geq 0}$, with $A_{\x,j}\in \mathbb{R}^{n}$, $b_{\x,j}>0$, $A_{\u,j}\in \mathbb{R}^{n}$, $b_{\u,j}>0$. We assume that $\mathcal{U}$ is compact.
In this paper, we provide a safety filter and safe set, which enable a safety certificate for arbitrary control signals $u_\mathcal{L}(t)\in\mathbb{R}^m$. 
\begin{definition}
\label{def:safe_set}
A set $\mathcal{S}\subseteq \mathcal{X}$ and a control law  $u_\mathcal{S}(x(t), u_\mathcal{L}(t),t)$, with $u_\mathcal{S}: \mathbb{R}^n \times \mathbb{R}^m \times \mathbb{N}_{\ge 0} \mapsto \mathcal{U}$ are called a safe set and a safety filter, respectively, if the application of the safety filter $u = u_\mathcal{S}$ in \eqref{eq:lin}, ensures robust constraint satisfaction, i.e., $x(t)\in \mathcal{X} ~\forall t\ge \bar t~\forall w(t)\in \mathcal{B}^n_\infty$, if $x(\bar t)\in \mathcal{S}$.
\end{definition}
We search for a safety filter $u_\mathcal{S}$ and a corresponding safe set $\mathcal{S}\subseteq\mathcal{X}$, such that state and input constraints satisfaction can be guaranteed for all future time steps and for all disturbances. 
While safety is the primary objective, the safety filter should also interfere as little as possible, i.e., whenever possible the safety filter should yield $u_{\mathcal{S}}=u_{\mathcal{L}}$. 
The main benefit of a modular safety filter according to Definition~\ref{def:safe_set} is that it can certify any potentially unsafe control input $u_\mathcal{L}$, including the important case of learning-based controllers or humans.

In the nominal case (i.e., without disturbance $w(t)$), the \ac{MPSF} formulation proposed in~\cite{Wabersich2018LinearControl} consists in solving the following optimisation problem at each time step%
\begin{subequations}\label{eq:nominal_MPSF}
    \begin{align}
    \min_{\Z, \V} \|\V_0 - u_\mathcal{L}(t)\|_2^2&
    \quad\text{s.t.}\quad & \Z_{k+1} = A\Z_k + B\V_k,~x(t) = \Z_0, ~k= 0,\dots, N-1,\label{eq:nom_pre}\\
    && \Z_k\in \mathcal{X},~\V_k\in \mathcal{U},~\Z_N\in \mathcal{X}_\f, ~k= 0,\dots, N-1,
    \end{align}
\end{subequations}
where $\Z_k,\V_k$ denote the nominal state and input, respectively, and $\mathcal{X}_\f\subseteq \mathcal{X}\subseteq \mathbb{R}^n$ is a positively invariant set as commonly required in \ac{MPC}, see e.g.~\cite{Kouvaritakis2016ModelStochastic}. 
In order to reduce the effects of model uncertainties, \acp{MPSF}~\cite{Wabersich2018LinearControl,Wabersich2021ASystems} employ robust \ac{MPC} techniques based on a so-called tube-controller, i.e., {$u(t) = v(t) + K(x(t) - z(t))$, where $K$ is fixed offline.}
In the following, we develop a robust counterpart of \eqref{eq:nominal_MPSF} based on \ac{SLS} (Section~\ref{sec:paramerisation}), where we not only optimise the (nominal) backup trajectory $\Z$ and $\V$, but also a tube-controller and disturbance reachable set for the error between the uncertain and nominal system. The resulting feedback policy guarantees robust constraint satisfaction and reduces conservativeness compared to other \acp{MPSF} (Section~\ref{sec:sl-mpc}). While offering a flexible approach, the required online optimisation can be computationally demanding. To address this limitation, we introduce an explicit system level safety filter (Section~\ref{sec:explicit}), {which removes most of the offline computational complexity}. Notably, both methods ensure robust constraint satisfaction.

\section{System Level Model Predictive Safety Filter}
\label{sec:implicit_safe_set}
In the following, we first present the parametrisation of the disturbance reachable sets under affine tube-controllers, before introducing the SL-\ac{MPSF} formulation.
\subsection{Parameterisation of affine controllers} 
\label{sec:paramerisation}
 The employed parameterisation of affine feedback policies for LTI systems relies on \ac{SLS} techniques~\cite{Goulart2006OptimizationConstraints,Anderson2019}. We define the prediction state $\X_k\in \mathbb{R}^n$ and input $\U_k\in \mathbb{R}^m$ for system~\eqref{eq:lin} as%
\begin{equation}
\begin{aligned}
        \X_{k+1} &= A \X_k + B \U_k + B_\w \W_k,~k= 0,\dots, N-1,
    \label{eq:nom_dyn}
\end{aligned}
\end{equation}
where $N\in\mathbb{N}_{\ge 1}$ is the prediction horizon and $\W_k \in \mathcal{B}_\infty^{n}={\mathcal{W}}$. 
We define the error (deviation) between the uncertain and nominal system states and inputs~\eqref{eq:nom_pre} with $\Delta \X_k \defmath \X_k - \Z_k$ and $\Delta \U_k \defmath \U_k - \V_k$.
In the following, we consider initial conditions of the form $\Delta \X_0 \in \mathcal{X}_0\defmath\Pz \mathcal{B}_\infty^{n}$, with $\Pz\in\mathbb{R}^{n \times n}$. This general initial set formulation will become essential for the explicit safe set (Section \ref{sec:explicit}), while the online optimisation-based SL-\ac{MPSF} (Section \ref{sec:sl-mpc}) only considers the special case of fixed initial conditions with $\Pz = 0_{n,n}$ and  $ \Delta \X_0 =0_n$. 
We can write the error dynamics compactly as%
\begin{equation}
    \Delta \X = \mathcal{Z}\mathcal{A}\Delta\X + \mathcal{Z}\mathcal{B}\Delta\U + \mathcal{E}\bm{\delta},
    \label{eq:LTV}
\end{equation}
with $\Delta \X = (\Delta \X_0, \dots, \Delta \X_N)$, $\Delta \U = (\Delta \U_0, \dots, \Delta \U_N)$, $\W\defmath (\W_0, \dots, \W_{N-1})$, $\bm{\delta} \defmath (\Delta \X_0, \W)\in \mathcal{X}_0 \times {\mathcal{W}}^{N}$, $\mathcal{A}\defmath \text{diag}(A, \dots, A, 0_{n,n})$, $\mathcal{B}\defmath \text{diag}(B,\dots, B, 0_{n,m})$, $\mathcal{E} \defmath \text{diag}(\mathcal{I}, B_\w, \dots, B_\w)\in \mathcal{L}^{N, n\times n},$ and the block downshift operator $\mathcal{Z}\in\mathcal{L}^{N,n \times n}$, i.e., a matrix
with identity matrices along its first block sub-diagonal and zeros elsewhere. 
We introduce the causal linear feedback controller $\Delta \U = \mathcal{K} \Delta \X$, with $\mathcal{K}\in \mathcal{L}^{N,m \times n}$, i.e.,%
\begin{equation}
    \U = \V + \mathcal{K} (\X - \Z), \label{eq:feedback}
\end{equation}
yielding $ \U_k = \V_k + \sum_{j=0}^{k} K^{k,j} (\X_{k-j} - \Z_{k-j})$. Using this feedback, we can write the closed-loop error dynamics as%
\begin{equation}
\begin{aligned}
\begin{bmatrix} \X \\\U \end{bmatrix}  - \begin{bmatrix} \Z \\\V \end{bmatrix} = \begin{bmatrix}
    \Delta \X\\\Delta \U
\end{bmatrix} = \begin{bmatrix} \mathcal{Z}(\mathcal{A} + \mathcal{B}\mathcal{K})\Delta \X + \mathcal{E}\bm{\delta}\\ \mathcal{K} \Delta \X \end{bmatrix}
=: \begin{bmatrix} {\Phi}_\x \\ {\Phi}_\u \end{bmatrix} \bm{\delta} =: \left[\begin{tabular}{c|c}
$\Phi_{\x,0}$ & $\Phi_{\x,\w}$\\
\hline
$ \Phi_{\u,0}$ & $ \Phi_{\u,\w}$\\
\end{tabular}
\right]
\left[
\begin{tabular}{c}
$\Delta \X_0$\\
\hline
$\W$
\end{tabular}\right],
\label{eq:sys_rep}
\end{aligned}
\end{equation}
where $\Phi_\x \in \mathcal{L}^{N,n\times n}$, $\Phi_\u \in \mathcal{L}^{N,m\times n}$.
The matrices ${\Phi}_{\x}$ and ${\Phi}_\u$ are called the system responses from the disturbances to the closed-loop error state and input, respectively. The following proposition shows that the closed-loop responses under arbitrary affine feedbacks lie on a linear subspace.

\begin{proposition}\cite[adapted from Theorem 1]{Chen2021SystemApproximation}
\label{prop:map}
Consider some disturbance sequence $\bm{\delta}\in \mathcal{X}_0 \times{\mathcal{W}}^{N}$. 
\begin{enumerate}[label=\alph*)]
    \item Any trajectory $\Delta \X$, $\Delta \U$ satisfying the dynamics~\eqref{eq:LTV},~\eqref{eq:feedback} also satisfy~\eqref{eq:sys_rep} with some ${\Phi}_\x \in \mathcal{L}^{N,n \times n}$, ${\Phi}_\u \in \mathcal{L}^{N,m \times n}$ lying on the subspace%
\begin{equation}
\left.
    \begin{aligned}
        &\left[ \mathcal{I} - \mathcal{Z}{\mathcal{A}}\quad -\mathcal{Z}{\mathcal{B}}\right] \left[\begin{array}{c}{\Phi}_{\x}\\{\Phi}_\u \end{array}\right] = \mathcal{E}.
    \end{aligned}
    \right.    
    \label{eq:affine_map}
\end{equation}
\item Let ${\Phi}_{\x}$ and ${\Phi}_\u$ be arbitrary matrices satisfying~\eqref{eq:affine_map}. Then the corresponding $\Delta \X$ and $\Delta \U$ computed with~\eqref{eq:sys_rep} also satisfy~\eqref{eq:LTV},~\eqref{eq:feedback} with ${\mathcal{K}} = {\Phi}_\u {\Phi}_{\x}^{-1}\in \mathcal{L}^{N,m\times n}$.
\end{enumerate}
\end{proposition}
\begin{remark}
\label{rem:ic_dist_size}
We consider square matrices $B_\w$, $\Pz$ to provide a simple exposition\footnote{For more general disturbance and initial condition sets, Proposition \ref{prop:map} requires $\mathcal{E}$ to be right-invertible, see e.g. \cite{Herold2022AControl}.}.
However, the results can be directly extended to more general disturbances and initial condition sets of the form $w(k)\in \mathcal{B}_\infty^{n_\w}$, $\mathcal{X}_0=P_0 \mathcal{B}_\infty^{n_0}$, compare, e.g.~\cite{Herold2022AControl}.
\end{remark}

For any nominal trajectory $\Z$, $\V$ satisfying \eqref{eq:nom_pre}, and any error feedback  $\Delta \U = \mathcal{K} \Delta \X$ with $\Phi_\x,\Phi_\u$ satisfying~\eqref{eq:affine_map}, the disturbance reachable sets of system~\eqref{eq:sys_rep} are given exactly by%
\begin{equation}
\begin{aligned}
    \X_k &\in \mathcal{R}_\x(\Z_k,\Phi_\x^k)  \defmath \{\Z_k\}\oplus \Phi_{\x,0}^{k}{\mathcal{X}}_0 \oplus \Phi_{\x,\w}^{k} \mathcal{B}_\infty^{N n},\\
    \U_k &\in \mathcal{R}_\u(\V_k,\Phi_\u^k) \defmath  \{\V_k\}\oplus \Phi_{\u,0}^{k}{\mathcal{X}}_0 \oplus \Phi_{\u,\w}^{k} \mathcal{B}_\infty^{Nn}.
\end{aligned}
\end{equation}

Hence, state and input constraint satisfaction can be robustly ensured by constraining the disturbance reachable sets to lie within the respective constraint sets, i.e., $\mathcal{R}_\x(\Z_k,\Phi_\x^k) \subseteq \mathcal{X}$ and $\mathcal{R}_\u(\V_k,\Phi_\u^k) \subseteq \mathcal{U}$ for all $k = 0, \dots, N-1$.
\begin{proposition}
\label{prop:cons}
There exists an affine feedback law of the form \eqref{eq:feedback} such that for any $\bm{\delta} \in \mathcal{X}_0 \times {\mathcal{W}}^{N}$,%
\begin{equation}
\begin{aligned}
        \X_k \in \mathcal{X} \defmath \{ x| A_{\x,j}x \le b_{\x,j},~j=1,\dots, n_\x\},
    \U_k \in \mathcal{U} \defmath \{ u| A_{\u,j}u \le b_{\u,j},~j=1,\dots, n_\u\},
\end{aligned}
\end{equation}
with $\X_k$ and $\U_k$ satisfying \eqref{eq:lin}, if and only if there exist matrices $\Phi_\x$ and $\Phi_\u$ and a nominal trajectory $\Z$, $\V$ satisfying \eqref{eq:affine_map}, \eqref{eq:nom_pre}, and, for $k=0,\dots N-1$, $j_\x = 1, \dots, n_\x$, $j_\u = 1, \dots, n_\u$:%
\begin{equation}
\begin{aligned}
\label{eq:cons_slc}
& A_{\x,j_\x} \Z_k +\|A_{\x,j_\x}\Phi_{\x,0}^{k}\Pz \|_1+\|A_{\x,j_\x}\Phi_{\x,\w}^{k} \|_1 \le b_{\x,j_\x},\\
&     A_{\u,j_\u} \V_k +\|A_{\u,j_\u}\Phi_{\u,0}^{k}\Pz \|_1+\|A_{\u,j_\u}\Phi_{\u,\w}^{k} \|_1 \le b_{\u,j_\u}.
\end{aligned}
\end{equation}
\end{proposition}
\begin{proof}
As per Proposition \ref{prop:map}, the state and input trajectory satisfy \eqref{eq:sys_rep}. Hence, for each $k=0,\dots,N-1$, $j = 1, \dots, n_\x$, we have%
\begin{equation}
\label{eq:proof_slc}
\begin{aligned}
\max_{\substack{\W \in {\mathcal{W}}^{N}\\ \Delta \X_0 \in {\mathcal{X}}_0}}A_{\x,j} \X_k 
&\stackrel{\eqref{eq:sys_rep}}{=} A_{\x,j} \Z_k + \max_{\Delta \X_0 \in {\mathcal{X}}_0} A_{\x,j} \Phi_{\x,0}^{k}\Delta \X_0 + \max_{\W \in {\mathcal{W}}^{N}}A_{\x,j}\Phi_{\x,\w}^{k}\W\\
&=A_{\x,j}\Z_k +\|A_{\x,j}\Phi_{\x,0}^{k}\Pz \|_1+\|A_{\x,j}\Phi_{\x,\w}^{k} \|_1 \le  b_{\x,j},
\end{aligned}
\end{equation}
where the second equality is given by the definition of the 1-norm. The same derivation applies to the input constraints. 
\end{proof}
As standard in \mbox{\ac{MPC}}, we ensure recursive feasibility by using a suitable terminal set~{\cite{Kouvaritakis2016ModelStochastic}}.
\begin{assumption}
\label{assum:terminal}
There exists a terminal set $\mathcal{X}_\f\defmath\{ x \in \mathbb{R}^n | A_\f x \le b_\f \}$ and a terminal feedback $K_\f\in\mathbb{R}^{m\times n}$ such that $(A+BK_\f) \mathcal{X}_\f \oplus B_\w \mathcal{W} \subseteq \mathcal{X}_\f$ and $(\mathcal{X}_\f\times K_\f\mathcal{X}_\f)\subseteq( \mathcal{X}\times \mathcal{U})$.
\end{assumption}
\subsection{{System level} predictive safety filter}
\label{sec:sl-mpc}
In this section, we present the SL-\ac{MPSF} using the parameterisation in Section \ref{sec:paramerisation}. 
In particular, at each time step $t$, for a given control input $u_\mathcal{L}(t)$ and a measured state $x(t)$, the SL-\ac{MPSF} is defined via the following optimisation problem:%
\begin{subequations}
\label{eq:sls}
    \begin{align}
    \min_{ \Phi,\Z,\V} \quad & \| \V_0 - u_\mathcal{L}(t)\|_2^2 ,\label{eq:cost} \\
    \text{s.t.}\quad   &\left[ \mathcal{I} - \mathcal{Z}\mathcal{A}\quad - \mathcal{Z}\mathcal{B}  \right]  \begin{bmatrix}
     \Phi_{\x}\\
     \Phi_\u\\
    \end{bmatrix} = \mathcal{E},~ x(t) = \Z_0,\label{eq:slp}\\
    & \Z_{k+1} = A\Z_k + B \V_k, ~k = 0,\dots, N-1,\label{eq:dyn}\\
    &  A_{\x,j} \Z_k +\|A_{\x,j}\Phi_{\x,\w}^{k} \|_1 \le b_{\x,j},~   j = 1,\dots, n_\x,~k = 0,\dots, N-1, \label{eq:cons1}\\
    & A_{\u,j} \V_k +\|A_{\u,j}\Phi_{\u,\w}^{k} \|_1 \le b_{\u,j},~ j = 1, \dots, n_\u, ~k = 0,\dots, N-1,\label{eq:cons2}\\
        &  A_{\f,j}\Z_N+\|A_{\f,j}\Phi_{\x,\w}^{N} \|_1 \le b_{\f,j},~ j = 1, \dots, n_\f.\label{eq:sls_terminal}
    \end{align}    
\end{subequations}
The constraints \eqref{eq:cons1}, \eqref{eq:cons2}, and \eqref{eq:sls_terminal} follow from Proposition~\ref{prop:cons}, where we consider the special case $\Pz = 0_{n,n}$ with $x(t)= \X_0 = \Z_0$. The terminal set constraint \eqref{eq:sls_terminal} ensures (robust) recursive feasibility.
The constraint \eqref{eq:slp} parameterises the tube-controller according to Proposition~\ref{prop:map}. The constraint \eqref{eq:dyn} defines the nominal dynamics used to compute a safe backup trajectory. The cost~\eqref{eq:cost} is chosen such that the applied input $u_\mathcal{L}(t)$ is modified as little as necessary, i.e., the solution of~\eqref{eq:sls} is $u_\mathcal{L}(t) = \V_0^\star$ whenever feasible.  Problem~\eqref{eq:sls} is a \ac{QP} and can hence be solved efficiently. 
We denote a minimiser of Problem~\eqref{eq:sls} by $\V^\star(x(t), u_\mathcal{L}(t))$.
In the closed-loop, we apply the first element of the optimal input sequence, i.e., the safety filter control law is given by%
\begin{equation}
    u_\mathcal{S}(x(t), u_\mathcal{L}(t),t) = \V_0^\star(x(t), u_\mathcal{L}(t)).
    \label{eq:safe_cont}
\end{equation}
We denote the set of feasible states $x(t)$ for Problem \eqref{eq:sls} by $\mathcal{X}_N \subseteq \mathcal{X}$.
\begin{theorem}
\label{theo:1}
Let Assumption~\ref{assum:terminal} hold. 
Then, the set $\mathcal{X}_N$ is a safe set according to Definition \ref{def:safe_set} for system \eqref{eq:lin} with the safety filter \eqref{eq:safe_cont}.
\end{theorem}
\begin{proof}
Suppose Problem~\eqref{eq:sls} is feasible for some state $x(t)\in\mathcal{X}_N$.
Then, for any $u_\mathcal{L}(t)\in \mathbb{R}^m$ and any $w(t)\in \mathcal{W}$, one can show that Problem~\eqref{eq:sls} is feasible for $x(t+1)$. Hence, the feasible set $\mathcal{X}_N$ is \ac{RPI}, i.e., $x(t+1) = Ax(t) + B\V_0^\star(x(t), u_\mathcal{L}(t)) + w(t) \in \mathcal{X}_N$. 
Following standard \ac{MPC} arguments, this can be shown by constructing a feasible candidate solution that shifts the previously optimal solution and appends the terminal control law $K_\f$ from Assumption~\ref{assum:terminal}. 
In fact, this proof can be found in~\cite{Goulart2006OptimizationConstraints}, where the result was shown based on the equivalent disturbance feedback formulation (cf.~\cite[Theorem 2]{Sieber2021AControl}). Notably, this recursive feasibility property is completely independent of the considered cost function and hence the input $u_{\mathcal{L}}(t)$.
\end{proof}
Compared to standard \ac{RMPC} formulations~\cite{Mayne2005RobustDisturbances,Chisci2001SystemsConstraints}, the proposed SL-\ac{MPSF} formulation~\eqref{eq:sls} reduces the conservativeness by optimising over affine feedbacks and using a tight reachability analysis (Proposition~\ref{prop:cons}). As a result, the SL-\ac{MPSF} scheme is less conservative than state-of-the-art \ac{MPSF} schemes, e.g.~\cite{Wabersich2018LinearControl}. 
As online optimisation can be computationally expensive, the next section introduces an alternative approach to synthesise a safe set and safety filter, which only requires the solution to a single LP offline. This drastically reduces the online computation complexity, albeit at the expense of potentially more frequent and stronger control interventions of the safety filter.

\section{Explicit system level safe set}
\label{sec:explicit}
The safety filter described in Section \ref{sec:implicit_safe_set} requires the solution of a \ac{QP} at each time step. {To eliminate} the need for an embedded \ac{QP}-solver, we present an explicit safe set that can be generated through the solution of a single LP, solved offline. The optimisation problem solved offline relies on the parameterisation introduced in Section \ref{sec:paramerisation}, and hence is comparable to Problem~\eqref{eq:sls}.

The key idea is that we maximise the size of the set of initial conditions $\mathcal{S}_\mathrm{e}$, such that the system is guaranteed to return to that same set $\mathcal{S}_\mathrm{e}$ within a given number of time steps, while always satisfying the constraints. Consequently, besides the constraint satisfaction given by Propositions~\ref{prop:map},~\ref{prop:cons}, this is realised by adding the following two constraints%
\begin{equation}
    \X_0 \in \mathcal{S}_\mathrm{e},~ \mathcal{R}_\x(\Z_N,\Phi_\x^N) \subseteq\mathcal{S}_\mathrm{e}.
    \label{eq:uncer_ic}
\end{equation} 
Previous approaches to design an (explicit) safe sets (cf., e.g.~\cite{Wabersich2018ScalableControl}), are based on a (robust) positively invariant set and offline controller design.%
In contrast, using system level disturbance reachable sets, the proposed safe set {does not require any controller tuning and} satisfies a weaker periodic invariance condition~\cite{GONDHALEKAR2011326}, and hence allows us to optimise over a simple shape of the safe set, e.g. unit infinity-norm balls.
In the following, we consider for simplicity the safe set $\mathcal{S}_\mathrm{e}\defmath  \{\Z_0\}\oplus \alpha \mathcal{B}_\infty^n$, with $\alpha > 0$, but hyperboxes, i.e., $\mathcal{S}_\mathrm{e}=\{\Z_0\}\oplus\text{diag}(\bm{\alpha})\mathcal{B}_\infty^{n}$, $\text{diag}(\bm{\alpha})\in \mathbb{R}^{n\times n}$ with $\bm{\alpha}\in\mathbb{R}^n$ could be considered directly. 
Note that Inequalities~\eqref{eq:cons_slc} in Proposition \ref{prop:cons} have a non-convex bilinearity if we optimise over the matrix $\Pz$. Hence, we use the change of variables $\tilde \Phi_\x^{k,k} \defmath \Phi_\x^{k,k}\Pz$, and $\tilde \Phi_\u^{k,k} \defmath \Phi_\u^{k,k}\Pz$ in \eqref{eq:cons_slc} for lossless convexification.
Due to the change of variables, the matrix $\mathcal{E}$ in \eqref{eq:affine_map} is replaced by $\tilde{\mathcal{E}}(\alpha) \defmath \text{diag}(\alpha \mathcal{I}, B_\w, \dots, B_\w)\in \mathcal{L}^{N, n\times n}$.

The following optimisation problem yields the proposed explicit safe set with an explicit safe backup control law:%
\begin{subequations}\protect
\label{eq:explicit_set}
    \begin{align}
    \min_{ \Phi,\Z,\V, \alpha} \quad &-\alpha,&
    \text{s.t.}\quad   &\left[ \mathcal{I} - \mathcal{Z}\mathcal{A}\quad - \mathcal{Z}\mathcal{B}  \right]  \begin{bmatrix}
     {\Phi}_{\x}\\
     \Phi_\u\\
    \end{bmatrix} = \tilde{\mathcal{E}}(\alpha),\label{eq:slp_expl}\\
    &&& \eqref{eq:slp},~\eqref{eq:cons1},~\eqref{eq:cons2}, \label{eq:cons_expl}\\
    &&&|\mathcal{I}_{n,j} (\Z_N-\Z_0)|+\| \mathcal{I}_{n,j}\Phi_{\x}^{N} \|_1 \le \alpha,~j =1,\dots, n,\label{eq:final_set_expl}
    \end{align}
\end{subequations}
where $\mathcal{I}_{n,j}$ is the $j^\text{th}$ row of $\mathcal{I}_n$. The constraints \eqref{eq:cons_expl} and \eqref{eq:final_set_expl} are similar from Problem~\eqref{eq:sls_terminal}.
We denote the optimal solution of the optimisation problem \eqref{eq:explicit_set} as $\Phi^\star$, $\Z^\star$, $\V^\star$ and $\alpha^\star$. The explicit safe set is given by $\mathcal{S}_\mathrm{e}^\star \defmath \{\Z_0^\star\}\oplus\alpha^\star \mathcal{B}_\infty^n$ and $\mathcal{K}^\star = (\Phi_\x^\star)^{-1} \Phi_\u^\star$ characterises a safe affine backup controller, similar to~\eqref{eq:feedback}\footnote{The constraint \eqref{eq:final_set_expl} guarantees that $\alpha^\star >0$ for $B_\w$ invertible and hence $\tilde{\mathcal{E}}(\alpha)$ is invertible.}. The objective function {of~\mbox{\eqref{eq:explicit_set}}} maximises the volume of the safe set. The parameterisation of the controller yields the constraints~\eqref{eq:slp_expl} as per Proposition~\ref{prop:map}. 
Note that the nominal trajectory $(\Z,\V)$ could also be chosen as zeros in \eqref{eq:explicit_set}, however, especially for asymmetrical constraints, this may introduce conservativeness.
The constraints \eqref{eq:final_set_expl} guarantee that the system will robustly return to the safe set within $N$ time steps under the safe backup controller (cf.~\eqref{eq:uncer_ic}). 
The constraints \eqref{eq:cons_expl} guarantee robust constraint satisfaction for any realisation of the disturbance and for any initial conditions given by \eqref{eq:uncer_ic}.
{Algorithm~\mbox{\ref{algo:1}} ensures safety by verifying that the learned-input \mbox{$u_\mathcal{L}(t)$} keeps the system within the safe set \mbox{$\mathcal{S}_\mathrm{e}^\star$}, which is computed using~\mbox{\eqref{eq:explicit_set}}, before applying it.}
\begin{algorithm}
\caption{Explicit system level safety filter}
Compute $\mathcal{S}_\mathrm{e}^\star$, $\mathcal{K}^\star$, $\Z^\star$, and $\V^\star$ using Problem~\eqref{eq:explicit_set}, Initialise $x(0) \in \mathcal{S}_\mathrm{e}^\star$\\
\For{$t = 0,1,2,\dots$}{
    \eIf{%
    \begin{equation}\label{eq:safe_control_law}
      u_\mathcal{L}(t) \in \mathcal{U}\land \{Ax(t)+Bu_\mathcal{L}(t)\}\oplus B_\w \mathcal{W} \subseteq \mathcal{S}_\mathrm{e}^\star
    \end{equation}}{
        Apply $u_{\mathcal{S}_\mathrm{e}}(t) = u_\mathcal{L}(t)$, $j \gets 0$
    }{
        Apply 
        \begin{equation}\label{eq:feedback_explicit}
            u_{\mathcal{S}_\mathrm{e}}(t) = \V_j^\star + \sum_{i=0}^j \mathcal{K}^{\star j,i} (x(t-i:t) - \Z_{i}^\star),~j \gets \text{modulo}(j+1,N)
        \end{equation} 
    }
}
\label{algo:1}
\end{algorithm}
We note that the condition \eqref{eq:safe_control_law} is a set containment condition, evaluated by only a few arithmetic operations thanks to the simple shape of the set $\mathcal{S}_\mathrm{e}^\star$.
\begin{theorem}
The control law resulting from Algorithm~\ref{algo:1} is a safety filter and the set $\mathcal{S}_\mathrm{e}^\star$ is the corresponding safe set according to Definition~\ref{def:safe_set}.
\end{theorem}
\begin{proof}
First, if the input $u_\mathcal{L}(t)$ is safe, i.e., $u_\mathcal{L}(t) \in \mathcal{U}\land \{Ax(t)+Bu_\mathcal{L}(t)\}\oplus B_\w \mathcal{W} \subseteq \mathcal{S}_\mathrm{e}^\star$, then, the state is guaranteed to stay within the set $\mathcal{S}_\mathrm{e}^\star$, by definition, and hence within the state constraints. Hence, the safe backup controller can always be applied subsequently to the application of a learned input. 

We now look at the case of a (potentially) unsafe input $u_\mathcal{L}(t)$. 
We denote the last time that the safety condition in \eqref{eq:safe_control_law} was fulfilled by $\bar{t}\in\mathbb{N}_{\geq 0}$. The safe backup controller constructed in \eqref{eq:explicit_set} is a direct application of Propositions \ref{prop:map} and \ref{prop:cons}. 
Hence, starting at $x(\bar t)\in \mathcal{S}_\mathrm{e}^\star$, it results that by applying the affine feedback \eqref{eq:feedback_explicit} for $N$ time steps to the system~\eqref{eq:lin}, the constraints
\begin{equation*}
    x(\bar t+k) \in \mathcal{X},~u(\bar t+k) \in \mathcal{U},~k = 0,\dots, N-1,
\end{equation*}
are robustly satisfied for any realisation of $\bm \delta \in \mathcal{X}_0 \times{\mathcal{W}}^{N}$, i.e., for any $x(\bar t) \in \alpha^\star \mathcal{B}_\infty^n$ and any sequence {\mbox{$(w(\bar t), \ldots, w(\bar t +N-1))\in \mathcal{W}^N$}}. Along the same lines, the constraint \eqref{eq:final_set_expl} guarantees that the system returns to the safe set within $N$ time steps, i.e., $x(\bar t+N) \in \mathcal{S}_\mathrm{e}^\star$. Hence the same backup controller is still feasible by re-setting $\bar t=t$.

It results that the backup controller associated to the safe set $\mathcal{S}_\mathrm{e}^\star$ is explicitly described by \eqref{eq:feedback_explicit}, and Algorithm \ref{algo:1} implements the associated safety filter as per Definition~\ref{def:safe_set}.
\end{proof}%
\begin{remark}
Based on the solution of \eqref{eq:explicit_set}, we could also construct an explicit \ac{RCI} set $\mathcal{C}^\star$ as%
\begin{equation}
    \mathcal{C}^\star \defmath \mathrm{Conv}\left(\bigcup_{k=0}^{N-1} \{\Z_k^\star\}\oplus \Phi_{\x}^{\star k}\mathcal{B}_\infty^{N n} \right),
\end{equation}
where $\mathrm{Conv}$ denotes the convex hull. The set $\mathcal{C}^\star$ defines a safe set as per Definition~\ref{def:safe_set}. Similar to \acp{CBF}~\cite{Cheng2019End-to-EndTasks}, the safety filter is implicitly defined by solving the following \ac{QP}%
\begin{subequations}
\label{eq:explicit_set_v2}
    \begin{equation}
    \min_{u\in\mathcal{U}} \|u - u_\mathcal{L}(t)\|_2^2,\quad \textrm{s.t.}~ Ax(t) + Bu + B_\w w \in \mathcal{C}^\star  ~\forall w\in \mathcal{B}_\infty^n.
    \end{equation}
\end{subequations}
\end{remark}
\section{Numerical example}
\label{sec:num}
In this section, we demonstrate the benefits of the proposed system level safety filters compared to state-of-the-art \ac{MPSF} formulations{ specifically highlighting improvements in the size of the safe set and the maximum amount of control intervention}.
We consider the illustrative example of controlling a double integrator system which relates to many practical problems, e.g. a spacecraft in rotation around one of its axis \cite{Mammarella2018Tube-basedDisturbance}. 
The dynamics are given by $
x(t+1) = \begin{bmatrix}
1&1\\0&1
\end{bmatrix}x(t) + \begin{bmatrix}
0.5\\
1
\end{bmatrix}u(t) + \begin{bmatrix}
0.3&0\\
0&0.3\\
\end{bmatrix} w(k),
$
with $|u(t)|\le 3$, $\|x(t)\|_\infty \le 5$, $\|w(t)\|_\infty \le 1$, $t\in\mathbb{N}_{\geq 0}$. Consider the linear feedback $K_\f$ computed via LQR design, with $Q=\mathcal{I}_2$, and $R=10^2$.
For this system and a prediction horizon of $N=10$, we compare the \ac{MPSF} safe set \mbox{$\mathcal{S}_\textrm{MPSF}$}~\cite{Wabersich2018LinearControl}, the maximal \mbox{\ac{RPI}} set $\Omega_\textrm{max}$, the SL-\ac{MPSF} safe set $\mathcal{S}_\textrm{SL-MPSF}$ (Section \ref{sec:sl-mpc}), the SL-based explicit safe set $\mathcal{S_\mathrm{e}}$ (Section \ref{sec:explicit}){ and the maximum {\ac{RCI}} set $\Xi_\textrm{max}$. Although $\Xi_\textrm{max}$ is not practically relevant because of its poor scalability, it serves as a basis for comparison as it is the maximum theoretically possible safe set}.
Using the MPT3 toolbox~\cite{mpt}, we find approximations of a minimal \mbox{\ac{RPI}} set {$\Omega_\textrm{min}$, maximal \mbox{\ac{RPI}} set $\Omega_\textrm{max}$, and maximal positively invariant set $\Pi_\textrm{max}$} based on the closed-loop dynamics $A+BK_\f$.
{We use the terminal set $ \mathcal{X}_\f=\Omega_\textrm{max}$ for the proposed SL-{\ac{MPSF}}.} For the \ac{MPSF} approach, we {use $\Omega_\textrm{min}$ for constraint tightening} and {$\Pi_\textrm{max}$ as the terminal set} (see e.g.~\cite{Mayne2005RobustDisturbances} for further details).

In Figure~{\ref{fig:fig1}}(right), the safe sets for each method are plotted and their size is compared against the largest possible safe set \mbox{$\Xi_\textrm{max}$}, with larger sets being preferable. The SL-{\ac{MPSF}}'s safe set is nearly maximal and significantly larger than the one associated with the {\ac{MPSF}}, demonstrating the effectiveness of our proposed method.
The explicit safe set covers regions not covered by either \mbox{$\Omega_\textrm{max}$} or \mbox{$\mathcal{S}_\textrm{MPSF}$}.
However, in general, the maximal {\ac{RPI}} set is not contained within the explicit safe set because of the simple scalar parameterisation of the safe set.

In Figure~\ref{fig:fig1}(left-middle), we show the maximum intervention for any input $u_\mathcal{L}(t)\in\mathcal{U}$ for each point in the safe set given by SL-\ac{MPSF} and \ac{MPSF}{, where lower values are preferable}.%
\begin{figure}[!htbp]
    \centering
    \includegraphics[width=\textwidth]{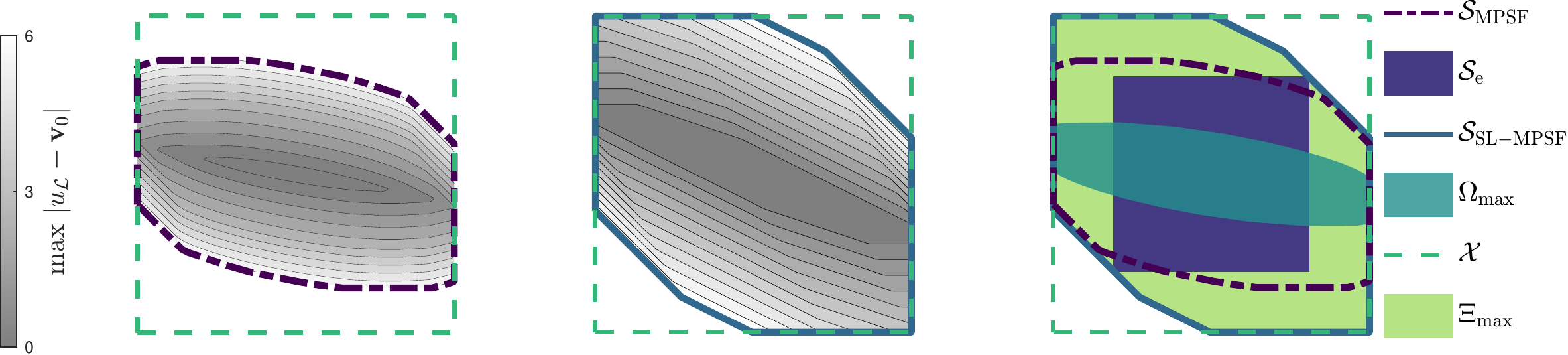}
    \caption{Comparison of the maximal control intervention for \ac{MPSF} (left) and SL-\ac{MPSF} (middle). Size comparison for safe set from \ac{MPSF} $\mathcal{S}_\textrm{MPSF}$, explicit $\mathcal{S}_\textrm{e}$, SL-MPSF $\mathcal{S}_\textrm{SL-MPSF}$, maximal \ac{RPI} set $\Omega_\textrm{max}$, maximal RCI set $\Xi_\textrm{max}$, and constraint set $\mathcal{X}$ (right).}
    \label{fig:fig1}
\end{figure}

The results highlight that SL-\ac{MPSF} is not only able to ensure safety for a larger set of states, but it also requires a less aggressive modification of potentially unsafe inputs. {This allows the learning-based controller to maintain high performance with minimal interference.} SL-\ac{MPSF} is hence the less conservative filter to explore the state space with minimal control interventions when the maximal \ac{RCI} cannot be computed.%
\paragraph{Computation times}
We uniformly sample \mbox{$10^4$} initial conditions in \mbox{$\mathcal{X}$}. The {\ac{MPSF}} solve time is \mbox{$4.8\textrm{ms}\pm 1.8\textrm{ms}$}, while SL-{\ac{MPSF}} solve time is \mbox{$65.7\textrm{ms} \pm 20.0\textrm{ms}$}, i.e., a factor $\approx 10$ slower. In contrast, the evaluation of the explicit safety filter (see Eq.~\mbox{\eqref{eq:safe_control_law}}) takes \mbox{$0.031\textrm{ms} \pm 0.028 \textrm{ms}$}, i.e., a factor $\approx 10^3$ faster compared to SL-{\ac{MPSF}}. However, this code has not been optimised for speed.\footnote{The code is available online \url{https://gitlab.ethz.ch/ics/SLS_safety_filter}. We use \texttt{YALMIP}~\mbox{{\cite{Lofberg2004}}} and \texttt{MOSEK}~\mbox{{\cite{mosek}}} on a  Intel\textsuperscript{\textregistered} Core\textsuperscript{TM} i7-8565U CPU @1.80GHz with 16.0GB of RAM.}

    
\section{Conclusion}
The paper has presented new predictive safety filters formulation based on {\ac{SLS}} techniques, to augment any control policy (including the important case of learning-based controllers) with safety guarantees. 
Two distinct safety filters have been proposed leveraging the flexibility of {\ac{SLS}}.
The first method extends {\ac{MPSF}} by enlarging the safe set through online optimisation over the robustness-ensuring tube-controller.
This approach allows to verify safety for a larger set of states with less control intervention, thereby reducing the impact on the performance of the potentially unsafe controller.
The second method presents an efficient explicit safety filter formulation that does not require solving optimisation problems online or tuning a controller offline. 
We expect that the proposed approaches can be extended to the nonlinear setting using the tools developped in~\mbox{\cite{leeman2023robust}}.
\bibliography{references}
\end{document}